\documentclass[12pt]{amsart}
\usepackage{amsmath,amssymb,amsfonts,amsopn,mathrsfs,bm}
\usepackage{mathtools}
\usepackage{graphicx}
\usepackage{epstopdf}
\usepackage{cite}
\DeclareGraphicsExtensions{.png}
\usepackage{caption}
\usepackage{subcaption}
\usepackage{tikz,pgfplots}
\tikzstyle{every node}=[font=\small]

\newcommand{\stft}{short-time Fourier transform}

\newtheorem{theorem}{Theorem}[section]

\newtheorem{corollary}[theorem]{Corollary}
\newtheorem{proposition}[theorem]{Proposition}
\newtheorem{definition}[theorem]{Definition}

\newtheorem{remark}[theorem]{Remark}

\theoremstyle{definition}

\newcommand{\beqa}{\begin{eqnarray*}}
\newcommand{\eeqa}{\end{eqnarray*}}

\newcommand{\field}[1]{\mathbb{#1}}
\newcommand{\bR}{\field{R}}        
\newcommand{\bN}{\field{N}}        
\newcommand{\bZ}{\field{Z}}        
\newcommand{\hlp}{{\tilde{L}^p_k}}

\def\cS{\mathcal{S}}

\def\cC{\mathcal{C}}

\def\a{\aleph}

\def\rd{\bR^d}

\def\rdd{{\bR^{2d}}}

\def\R{\Big))}

\def\<{\left<}
\def\>{\Big)>}

\def\inv{^{-1}}

\def\mv1{M_v^1}

\def\phas{(x,\xi)}
\def\mn{(m,n)}
\def\mn'{(m',n')}

\def\o{\eta}
\def\a{\alpha}

\def\R{\mathbb{R}}
\def\Ren{\mathbb{R}^d}

\def\tauhz0{\widehat{\mathcal{T}}^\hbar(z_0)}
\def\tauhz{\widehat{\mathcal{T}}^\hbar(z)}

\def\Sn2{S_{2}(L^{2}(\Ren))}
\def\S1{S_{1}(L^{2}(\Ren))}
\def\sig00{\sigma_{0,0}}

\begin{document}
\begin{abstract}
\medskip We consider a class of Schr\"odinger equations with time-dependent smooth magnetic and electric potentials having a growth at infinity at most linear and quadratic, respectively. We study the convergence in $L^p$ with loss of derivatives, $1<p<\infty$, of the time slicing approximations of the corresponding Feynman path integral. The results are completely sharp and hold for long time, where no smoothing effect is available.
The techniques are based on the decomposition and reconstruction of functions and operators with respect to certain wave packets in phase space. 
\medskip
\end{abstract}

\title{Convergence in $L^p$ for Feynman path integrals}

\author{Fabio Nicola}
\address{Dipartimento di Scienze Matematiche,
Politecnico di Torino, corso Duca degli Abruzzi 24, 10129 Torino,
Italy}
\email{fabio.nicola@polito.it}

\maketitle

\section{Introduction}

Feynman path integrals were introduced in 1948 \cite{feynman,feynman-higgs} to provide a new formulation of Quantum Mechanics and nowadays represent a fundamental tool in most branches of modern Physics. In particular, R. Feynman suggested the construction of the integral kernel $K(t,s,x,y)$ of the Schr\"odinger propagator as a suggestive sum-over-histories, in the following sense. First of all the kernel $K(t,s,x,y)$ itself is interpreted as the probability amplitude for a particle to be at the point $x$ at time $t$ provided it was at $y$ at time $s$ ($x,y\in\rd$). Now, in the computation of this quantity every path $\gamma$ joining $y$ and $x$, therefore satisfying $\gamma(s)=y$, $\gamma(t)=x$, carries a contribution which is proportional to $e^{i\hbar^{-1}S[\gamma]}$, where $S[\gamma]$ is the action along the path $\gamma$:
\[
S[\gamma]=\int_s^t L(\gamma(\tau),\dot{\gamma}(\tau),\tau)\,d\tau,
\]
$L$ being the Lagrangian of the corresponding classical system.
The total amplitude is finally obtained by superposition and can be written symbolically as an integral
\[
K(t,s,x,y)=\int e^{i\hbar^{-1}S[\gamma]}\mathcal{D}[\gamma]
\]
over the space of paths satisfying the above boundary conditions, where $\mathcal{D}[\gamma]$ is a suitable measure on this space. Although such a measure actually does not exist in the measure theoretic sense \cite{cameron}, several rigorous justifications have been proposed by many authors and from different viewpoints (analytic continuation of the parabolic propagator, infinite dimensional oscillatory integrals, stochastic integrals, etc.). The literature is enormous and we refer to the books \cite{albeverio,mazzucchi,reed,schulman} and the references therein. Instead here we focus on the original approach of Feynman \cite{feynman,feynman-higgs} via time slicing approximations, which was carried on in a rigorous way in the papers\cite{fujiwara1,fujiwara2,fujiwara3,fujiwara4,fujiwara5,ichinose1,ichinose2,ichinose3,ichinose4,kitada1,kitada2,kumanogo,kumanogo2,kumanogo3,tsuchida,yajima}. Briefly one argues as follows. Suppose that for $|t-s|$ small enough there is only one classical path $\gamma$ (i.e.\ a path satisfying the Euler-Lagrange equations) satisfying the boundary condition $\gamma(s)=y$, $\gamma(t)=x$. Define then the action
\begin{equation}\label{azione}
S(t,s,x,y)=\int_s^t L(\gamma(\tau),\dot{\gamma}(\tau),\tau)\,d\tau,
\end{equation}
along that path.\par
Consider the operator $E^{(0)}(t,s)$ defined by
\begin{equation}\label{ezero}
E^{(0)}(t,s)f(x)=\frac{1}{(2\pi i (t-s) \hbar)^{d/2}} \int_{\rd} e^{i\hbar^{-1}S(t,s,x,y)} f(y)\, dy.
\end{equation}
The idea is that this operator should represent a good approximation of the actual propagator when $|t-s|$ is small (in fact, for the free particle $E^{(0)}(t,s)$ coincides with the exact propagator). In general one then considers a subdivision $\Omega:s=t_0<t_1<\ldots<t_L=t$ of the interval $[s,t]$ and the composition
\begin{equation}\label{zero0}
E^{(0)}(\Omega,t,s)=E^{(0)}(t,t_{L-1}) E^{(0)}(t_{L-1},t_{L-2})\ldots E^{(0)}(t_1,s),
\end{equation}
which has integral kernel
\begin{multline}\label{zero}
K^{(0)}(\Omega,t,s,x,y)
\\
=\prod_{j=1}^L \frac{1}{(2\pi i(t_j-t_{j-1})\hbar)^{d/2}}\int_{\R^{d(L-1)}} \exp\Big(i\hbar^{-1}\sum_{j=1}^L S(t_j,t_{j-1},x_j,x_{j-1})\Big) \prod_{j=1}^{L-1} dx_j,
\end{multline}
with $x=x_{L}$ and $y=x_{0}$; see Figure 1.\par
\begin{figure}[b]
\begin{tikzpicture}[yscale=0.5, xscale=1.2]
\draw[-stealth] (-0.2, 0) -- (5.5, 0) node[above] {}; 
\draw[-stealth] (0, -0.4) -- (0, 5.5) node[below right] {}; 
\draw[] (1,0) node[below]{$s=t_0$} -- (1,5);
\draw[] (2,0) node[below]{$t_1$} -- (2,5);
\draw[] (3,0) node[below]{$t_2$} -- (3,5);
\draw[] (4,0) node[below]{$$} -- (4,5);
\draw[] (5,0) node[below]{$t=t_L$} -- (5,5);
\draw[very thick] (1,1) to [out=80,in=200] (2,2);
\draw[very thick] (2,2) to [out=70,in=240] (3,4);
\draw[very thick] (3,4) to [out=-70,in=145] (4,1);
\draw[very thick] (4,1) to [out=80,in=195] (5,3);
\draw[dashed] (0, 1) node[left]{$y=x_0$} -- (1, 1);
\draw[dashed] (0, 2) node[left]{$x_1$} -- (2, 2);
\draw[dashed] (0, 3) node[left]{$x=x_L$} -- (5, 3);
\end{tikzpicture}
\caption{A piecewise classical path in spacetime.}
\end{figure}
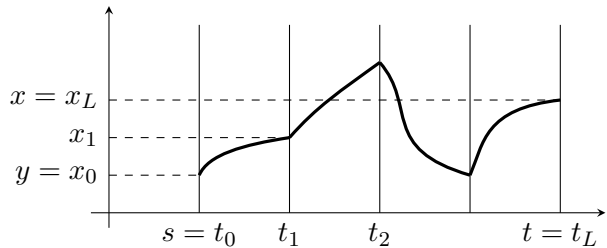

Feynman stated that the operator $E^{(0)}(\Omega,t,s)$ converges (in some sense) to the actual propagator as $\omega(\Omega):=\sup_{j=1,\ldots,L}|t_j-t_{j-1}|\to 0$. In view of the semiclassical approximation one also introduces higher order parametrices $E^{(N)}(t,s)$, $N=1,2\,\ldots$, and the corresponding time slicing approximations $E^{(N)}(\Omega,t,s)$ (see Section 2 below), and set the same convergence problem.  In the present paper we investigate this issue for the uniform topology of bounded operators on $L^p$-based Sobolev spaces.\par 
We consider the Schr\"odinger equation
\begin{equation}\label{equazione}
i\hbar \partial_t u=-\frac{1}{2}\hbar^2\Delta u+V(t,x)u
\end{equation}
where $0<\hbar\leq 1$ and the potential $V(t,x)$, $t\in \R$, $x\in\rd$, satisfies the following hypothesis.\par\medskip
{\bf Assumption (A)} 
{\it $V(t,x)$ is a real function of $(t,x)\in \R\times\rd$ with $\partial^\alpha_x V(t,x)$ continuous in $(t,x)\in \R\times \rd$, for every $\alpha\in\bN^d$ and satisfying
\[
|\partial^\alpha_x V(t,x)|\leq C_\alpha,\quad|\alpha|\geq 2,\ (t,x)\in \R\times \rd.
\]}
\par\smallskip
Under this hypothesis the exact propagator $U(t,s)$ was constructed in \cite{fujiwara1,fujiwara2}: it turns out that, {\it for $|t-s|$ small enough}, $U(t,s)$ is an oscillatory integral operator of the form 
\begin{equation}\label{cinque}
U(t,s)f(x)=\frac{1}{(2\pi i (t-s) \hbar)^{d/2}} \int_{\rd} e^{i\hbar^{-1}S(t,s,x,y)} b(\hbar,t,s,x,y)f(y)\, dy
\end{equation}
for some amplitude $b$ such that $\partial^\alpha_x\partial^\beta_y b(\hbar,t,s,x,y)$ is bounded, for every $\alpha,\beta\in\bN^d$, $0<\hbar\leq1$. Instead, for $|t-s|$ large, $U(t,s)$ can be written as the composition of a finite number of oscillatory integral operators as above.\par Moreover $E^{(N)}(\Omega,t,s)$ (for fixed $N$) was also shown to converge in a quite strong topology to the actual propagator as $\omega(\Omega)\to0$ when $t-s$ is small enough, which implies the convergence in the uniform topology of bounded operators on $L^2(\rd)$. While we are not aware of similar results in $L^p(\rd)$ for $p\not=2$, it is clear that in that case the propagator is in general not even bounded on $L^p$ but a loss of derivatives occurs. For example, for the free propagator we have
\[
\| e^{i\hbar \Delta}f\|_{L^p}\leq C\| (1-\hbar\Delta)^{k/2}f\|_{L^p},\quad k=2d|1/2-1/p|,\ 1<p<\infty.
\]
This estimate was proved in \cite{miyachi} in the case $\hbar=1$, whereas the general case follows at once by a scaling argument. Generalizations to different classes of potentials were proved, except for the endpoint, in \cite{Blu07,CCO,JN95,JN94,nicola}, exploiting in an essential way some smoothing effect (Gaussian estimates for the heat kernel). The above loss of derivatives is optimal \cite{brenner} and is explained by the fact that the characteristic manifold of the Schr\"odinger operator, that is a paraboloid, has non-vanishing Gaussian curvature. For hyperbolic equations \cite{ferreira,seeger-sogge-stein,stein93} we have instead the loss $k=2(d-1)|1/2-1/p|$, because there is one flat direction in that case. \par
Motivated by this model, we will focus on the convergence of the time slicing approximations in the following Sobolev spaces.\par
 For $1<p<\infty$, $k\in\R$, define the space
\begin{equation}\label{spaziohlp}
\hlp=\hlp(\rd)=\{f\in\cS'(\rd): \|f\|_{\hlp}=\|(1-\hbar\Delta)^{k/2}f\|_{L^p}<\infty\}.
\end{equation}
Hence, $\hlp$ is the usual Sobolev space as a vector space, but with a rescaled norm at the Planck scale. Here is our main result. 
\begin{theorem}\label{mainteo}
Assume the condition in Assumption {\rm (A)}. Let $1<p<\infty$, $k=k_p=2d|1/2-1/p|$. 
\begin{itemize}
\item[a)] For every $T>0$ there exists a constant $C(T)>0$ such that, for all $f\in\cS(\rd)$, $|s-t|\leq T$, $0<\hbar\leq 1$,
\begin{equation}\label{uno}
\|U(t,s)f\|_{L^p}\leq C(T) \|f\|_{\hlp},\quad 1<p\leq 2,
\end{equation}
\begin{equation}\label{due}
\|U(t,s)f\|_{\tilde{L}^p_{-k}}\leq C(T) \|f\|_{L^p},\quad 2\leq p<\infty.
\end{equation}
\item[b)] For every $T>0$, $N=0,1,2,\ldots$,  there exists a constant $C(T)>0$ such that, for $0<t-s\leq T$ and any subdivision $\Omega$ of the interval $[s,t]$, $f\in\cS(\rd)$, $0<\hbar\leq 1$, we have 
\begin{multline}\label{tre}
\|\big(E^{(N)}(\Omega,t,s)-U(t,s)\big)f\|_{L^p}
\leq C(T)\hbar^N \omega(\Omega)^{N+1}|t-s| \|f\|_{\hlp},\ 1<p\leq 2,
\end{multline}
\begin{multline}\label{quattro}
\| \big(E^{(N)}(\Omega,t,s)-U(t,s)\big)f\|_{\tilde{L}^p_{-k}}
\leq C(T)\hbar^N \omega(\Omega)^{N+1}|t-s|  \|f\|_{L^p},\ 2\leq p<\infty.
\end{multline}
\end{itemize}
\end{theorem}
We will show in the last section that these are, in general, all possible $L^p$-estimates for Schr\"odinger propagators of the above type. Extensions to the case of magnetic fields (cf. \cite{yajima,fujiwara5,tsuchida}) or even rough potentials will be discussed below.\par
Let us now say a few words about the strategy of the proof. There are two main issues. First, for {\it long time} the propagator does not have the integral representation \eqref{cinque} anymore and in general no smoothing effect is available. For example, for the quadratic potential $V(t,x)=|x|^2/2$ (and $\hbar=1$ for simplicity), the propagator $U(t,0)$ at time $t=k\pi$, $k\in\bZ$, has distribution kernel $e^{-ik\pi/2}\delta((-1)^k x-y)$ (see e.g.\ \cite{rod}). Secondly, the space of bounded operators $\hlp\to L^p$ is of course not an algebra, and hence it does not behave nicely with respect to compositions as those in \eqref{zero0}. To overcome these problems we lift the whole analysis to the phase space, using ideas originated by the seminal papers \cite{fefferman,marzuola,kt,tataru} and developed by the author et al.\ in \cite{cnr,cgnr,cnr1}. The key tool is a space of operators highly sparse with respect to certain wave packet decompositions. To be precise, consider a fixed Schwartz function $g\not\equiv0$, and the wave packet 
\[
\pi(z)g(y)=e^{i y\xi} g(y-x),\quad z=(x,\xi)\in\rd\times\rd
\]
which is highly concentrated, in phase space, near $(x,\xi)$. We will consider operators $T$ satisfying (when $\hbar=1$)
\[
|\langle T\pi(z)g,\pi(w)g\rangle|\leq C_m (1+|w-\chi(z)|)^{-m},\quad z,w\in\rdd,
\]
for every $m\in\bN$ and for some canonical transformation $\chi:\rdd\to\rdd$.
Composition of operators corresponds to the composition of the corresponding maps $\chi$'s. Moreover these operators are trivially bounded on the so-called modulation spaces $M^p$, $1\leq p\leq\infty$ (see \cite{deGosson,book} and Section 3 below), which measure the phase space concentration of a function. We finally come back to $L^p$ spaces by means of a nontrivial embedding recently obtained in  \cite{ks}. \par
 
While {\it endpoint} continuity results for $p\not=2$ typically require a painful analysis in some version of the Hardy space (see e.g. \cite{fefferman,miyachi,seeger-sogge-stein,stein93}), one of the novelties of this approach is that such technicalities are hidden under the above mentioned embedding and the analysis is focused on how the operators move the wave packets in phase space, which is much smoother and closer to intuition. For this reason and in view of the results in \cite{cgnr2,cnr10,marzuola}, we believe that the present approach could also extend to the case of rough potentials (and degenerate operators). This would provide an interesting low regularity version of the results of \cite{fujiwara2,yajima}. We plan to investigate these issues in a subsequent paper (see also the last section for some partial result in this connection).\par
Briefly, the paper is organized as follows. In Section 2 we recall the short-time analysis of \cite{fujiwara1,fujiwara2,yajima}. Section 3 is devoted to the phase space methods used in the sequel and, in particular, to the class of operators, closed under composition, mentioned above. In Section 4 we will show that the short-time propagators $U(t,s)$ and approximations $E^{(N)}(\Omega,t,s)$ belong to such operator class. Theorem \ref{mainteo} will be proved in Section 5. Section 6 treats a generalization concerning the convergence of  $E^N(\Omega,t,s)$ with $N\geq 1$ in the presence of a magnetic field satisfying the same assumptions as in \cite{yajima}. Finally in Section 7 we show the sharpness of the results and provide an extension of the first part of Theorem \ref{mainteo} to a class of rough potentials.
  
\section{The Schr\"odinger flow: short time analysis}
In this section we recall the main results obtained in \cite{fujiwara1,fujiwara2}; see also \cite{yajima}. We will always assume the hypothesis in Assumption (A) in Introduction.\par
First of all we fix the following notation. We denote by $S^0_{0,0}$ the class of smooth function $a(x,y)$ in $\rdd$ which are bounded together with their derivatives, endowed with the seminorms
\[
\|a\|_m=\sup_{|\alpha|+|\beta|\leq m}\|\partial^\alpha_x\partial^\beta_y a\|_{L^\infty(\rdd)},\quad m\in\bN.
\]
Consider now the Hamiltonian 
\[
H(t,x,\xi)=\frac{1}{2}|\xi|^2+V(t,x).
\] 
Denote by $(x(t,s,y,\eta),\xi(t,s,y,\eta))$ ($s,t\in\R$, $y,\eta\in\rd$), the solution of the corresponding system 
\[
\dot{x}=\xi,\quad \dot{\xi}=-\nabla_x V(t,x)
\]
with initial condition at time $t=s$ given by $x(s,s,y,\eta)=y$, $\xi(s,s,y,\eta)=\eta$. The flow
\begin{equation}\label{chi}
(x(t,s,y,\eta),\xi(t,s,y,\eta))=\chi(t,s)(y,\eta)
\end{equation}
defines a smooth canonical transformation $\chi(t,s):\rdd\to\rdd$ satisfying for every $T>0$ the estimates 
\begin{equation}\label{stimeflusso}
|\partial^\alpha_y\partial^\beta_\eta x(t,s,y,\eta)|+|\partial^\alpha_y\partial^\beta_\eta \xi(t,s,y,\eta)|\leq C_{\alpha,\beta} (T)\quad y,\eta\in\rd,
\end{equation}
for some constant $C(T)>0$, if $|t-s|\leq T$ (see \cite[Proposition 1.1]{fujiwara2}).\par
Moreover, there exists $\delta>0$ such that for $|t-s|<\delta$ and every $x,y\in\rd$, there exists only one solution such that $x(s)=y$, $x(t)=x$. By computing the action along this path, as in \eqref{azione}, we define the generating function $S(t,s,x,y)$ for $|t-s|<\delta$. It satisfies the estimates
\begin{equation}\label{pre1}
|t-s|\,|\partial^\alpha_x\partial^\beta_y S(t,s,x,y)|\leq C_{\alpha,\beta},\quad |\alpha|+|\beta|\geq 2
\end{equation}
and
\begin{equation}\label{pre2}
|t-s|\,\Big|{\rm det}\,\Big(\frac{\partial^2 S(t,s,x,y)}{\partial y^2}\Big)\Big|\geq \tilde{\delta},
\end{equation}
for some $\tilde{\delta}>0$ and every $x,y\in\rd$, always for $|t-s|<\delta$.\par
We now come to the construction of the parametrices. For $|t-s|<\delta$, define the operator $E^{(0)}(t,s)$ for $|t-s|<\delta$ as in \eqref{ezero}.\par
 For $N=1,2,\ldots$ we define $E^{(N)}(t,s)$ as
\begin{equation}
E^{(N)}(t,s)f(x)=\frac{1}{(2\pi i (t-s) \hbar)^{d/2}} \int_{\rd} e^{i\hbar^{-1}S(t,s,x,y)} e_N(\hbar,t,s,x,y) f(y)\, dy,
\end{equation}
with $e_N(\hbar,t,s,x,y)=\sum_{k=1}^N (i\hbar)^{k-1} a_k(t,s,x,y)$, where the amplitudes $a_k$ satisfy the transport equations
\[\label{trasporto}
\frac{\partial a_k}{\partial t}+\sum_{j=1}^d\frac{\partial S}{\partial x_j}\frac{\partial a_k}{\partial x_j}
+\frac{1}{2}\Big(\Delta_x S-\frac{d}{t-s}\Big)a_k=\frac{1}{2}\Delta_x a_{k-1},
\]
with $a_{0}(t,s,x,y)\equiv0$ and the initial conditions $a_1(s,s,x,y)=1$ and $a_k(s,s,x,y)=0$ for $k=1,2,\ldots$\par 
We have, for every $m\in\bN$,
\begin{equation}\label{quattro0}
\|a_k(t,s,\cdot,\cdot)\|_m\leq C_m\quad {\rm for}\ |t-s|<\delta.
\end{equation}
The operators $E^{(N)}(t,s)$ are parametrices in the sense that, for $N=0,1,\ldots,$
\[
\big(i\hbar\partial_t+\frac{1}{2}\hbar^2\Delta-V(t,x)\big) E^{(N)}(t,s)f=G^{(N)}(t,s)f
\]
with
\begin{equation}\label{defgn}
G^{(N)}(t,s)f =\frac{1}{(2\pi i (t-s) \hbar)^{d/2}} \int_{\rd} e^{i\hbar^{-1}S(t,s,x,y)} g_N(\hbar,t,s,x,y) f(y)\, dy,
\end{equation}
where $g_N$ satisfies the estimates \cite[Propositions 1.5, 1.6]{fujiwara2}
\begin{equation}\label{stimeresto}
\|g_N(\hbar,t,s,\cdot,\cdot)\|_m\leq C_m\hbar^{N+1}|t-s|^{N+1}.
\end{equation}

Moreover, for a subdivision $\Omega:s=t_0<t_1<\ldots<t_L=t$ we define
\begin{equation}\label{zero0}
E^{(N)}(\Omega,t,s)=E^{(N)}(t,t_{L-1}) E^{(N)}(t_{L-1},t_{L-2})\ldots E^{(N)}(t_1,s).
\end{equation}
As already observed in the introduction for $|t-s|<\delta$ the propagator has the form
\begin{equation}\label{cinque0}
U(t,s)f(x)=\frac{1}{(2\pi i (t-s) \hbar)^{d/2}} \int_{\rd} e^{i\hbar^{-1}S(t,s,x,y)} b(\hbar,t,s,x,y)f(y)\, dy
\end{equation}
 for an amplitude $b$ such that $\partial^\alpha_x\partial^\beta_y b(\hbar,t,s,x,y)$ is of class $C^1$ in $t,s,x,y$ and satisfying
\begin{equation}\label{cinque1}
\|b(\hbar, t,s,\cdot,\cdot)\|_m\leq C_m
\end{equation}
for $|t-s|<\delta$, $0<\hbar\leq 1$, $m\in\bN$.\par
When $|t-s|$ is large, $U(t,s)$ is the composition of a finite number of such oscillatory integral operators. 

\section{Phase space methods}
We recall here some basic facts of phase space analysis used in the sequel; see \cite{book} and especially \cite{deGosson} for a Mathematical Physics perspective. Then we report on some results from \cite{cgnr,cnr1,tataru}.\par

\subsection{Phase space concepts}
For a point $z=\phas\in\rd\times\rd$ and a function $f$ on $\rd$,
we denote the phase-space
shifts by $$\pi(z)f(y)= e^{i y\xi} f(y-x),\quad
\mbox{where}\quad y\xi=y\cdot\xi=\sum_{k=1}^d y_k\xi_k \, .$$

The \stft\, (STFT), or Bargmann transform, or FBI transform, of a function/distribution $f$ on $\rd$ with respect to a Schwartz window function $g\in\cS(\rd)\setminus\{0\}$ is defined by

\begin{equation}\label{stft}
V_g f(z)=\langle f,\pi(z)g
\rangle= \int _{\rd} f(y)\overline{g(y-x)}e^{- iy\xi}\,dy,
\end{equation}
for $z=(x,\xi) \in\rd\times \rd$.
Its adjoint is given by
\[
V^\ast_g F(y)=\int_{\rdd} F(z) \pi(z) g(y)\, dz,\quad y\in\rd,
\]
and we have the inversion formula (see e.g.
 \cite[Corollary 3.2.3]{book})
\begin{equation}\label{inversione}
f=\frac{1}{(2\pi)^{d}\|g\|^2_{L^2}}V_g^\ast V_g f=
\frac{1}{(2\pi)^{d}\|g\|^2_{L^2}}\int_{\rdd} V_g f(z)\pi(z)g\,dz.
\end{equation}

For $1\leq p\leq\infty$ the modulation space $M ^{p}=M^p(\rd)$ is the space of
distributions $f\in\cS'(\rd)$ such that their STFTs belong to the
space $L^{p}(\rdd) $ with  norm
\[
\|f\|_{M^p(\R^d)}:=\|V_g f\|_{L^{p}(\R^{2d})}.
\]
This definition  does not depend on the choice of the window $g\in \cS (\rd ), g \neq 0$, and different windows yield equivalent norms on $M^p$~\cite[Theorem 11.3.7]{book}.
We also have 
\[
V_g: M^p(\rd)\to L^p(\rdd),\quad V^\ast_g: L^p(\rdd)\to M^p(\rd)
\]
as bounded operators.\par
The following embedding results (see \cite[Theorems 1.3, 1.4]{ks} and also \cite{tomita}) are crucial in the sequel.\par
  For $1<p<\infty$, $k\in\R$, let $L^p_k=L^p_k(\rd)$ be the space of
distributions $f\in\cS'(\rd)$ such that $(1-\Delta)^{k/2} f\in L^p$, with the norm
\[
\|f\|_{L^p_k}=\|(1-\Delta)^{k/2} f\|_{L^p}.
\]
\begin{theorem}\label{proimmersioni} Let $1<p<\infty$ and $k=k_p=2d|1/2-1/p|$. Then we have 
\begin{equation}\label{immersioni}
M^p\hookrightarrow L^p,\quad L^p_k\hookrightarrow M^p,\quad 1<p\leq 2
\end{equation}
as well as 
\[
L^p\hookrightarrow M^p,\quad M^p \hookrightarrow L^p_{-k}\quad 2\leq p<\infty.
\]
\end{theorem}

\subsection{Canonical transformations and the associated operators}
We now introduce a special class of operators, defined in terms of their sparsity with respect to Gabor wave packets and a given canonical transformation.\par
A canonical transformation $(x,\xi)=\chi(y,\eta)$ will be called {\it tame} if it enjoys the following two properties:
\par\medskip
\noindent {\it  A1.} $\chi:\rdd\to\rdd$ is smooth, invertible,  and
preserves the symplectic form in $\rdd$, i.e., $dx\wedge d\xi= d
y\wedge d\eta$; $\chi $ is a \emph{symplectomorphism}.
\par\medskip\noindent
{\it  A2.} We have
\begin{equation}\label{chistima}
|\partial^\alpha_y\partial^\beta_\eta \chi(y,\eta)|\leq C_{\a,\beta},\quad |\a|+|\beta|\geq 1,\ y,\eta\in\rd.\end{equation}

For example, the canonical transformation $\chi(t,s)$ in \eqref{chi} satisfies {\it A1} and {\it A2}, with constants $C_{\alpha,\beta}=C_{\alpha,\beta}(T)$, provided $|t-s|\leq T$ ($T>0$ being arbitrary), in view of \eqref{stimeflusso}. 

Let us observe that {\it A1} and {\it A2} imply that $\chi$ and
$\chi^{-1}$ are globally Lipschitz. With the notation $\langle z\rangle=(1+|z|^2)^{1/2}$ (Japanese bracket), this property implies that
$$
C^{-1}\langle w-\chi (z) \rangle\leq \langle \chi \inv (w) - z\rangle\leq C \langle w-\chi (z) \rangle
\qquad w,z\in \rdd \, ,
$$
for some constant $C>0$ {\it depending only on an upper bound for the first derivatives of $\chi$}. 
The following class of operators was introduced in \cite{tataru}.

\begin{definition}
Let $\chi$ be a tame canonical transformation. Let $g\in\cS(\rd)\setminus\{0\}$.  We denote by $FIO(\chi)$ the class of operators $T:\cS(\rd)\to \cS'(\rd)$ such that, for every\footnote{
In this definition we could of course take $m$ integer; we take $m$ real in view of possible extensions to the case of rough potentials, where the propagators should  satisfy the same estimates as in \eqref{std} but for a {\it fixed} $m>2d$. In that case it is important to allow $m$ to be as small as possible, in order to reach minimal regularity. 
} $m>0$,
\begin{equation}\label{std}
|\langle T\pi(z)g,\pi(w) g\rangle|\leq C_m\langle w-\chi(z)\rangle^{-m},\quad z,w\in\rdd.
\end{equation}
We endow this space with the seminorms 
\[
\|T\|_{m,\chi}=\sup_{z,w\in\rdd} \langle w-\chi(z)\rangle^{m}|\langle T\pi(z)g,\pi(w) g\rangle|.
\]
\end{definition}

It is proved in \cite{cgnr} that the definition of $FIO(\chi)$ does not depend on the window $g\in\cS(\rd)\setminus\{0\}$.\par
 The following two theorems can also be found in \cite{cgnr}, but we provide here the proof for the benefit of the reader, because they were proved there in a discrete framework; moreover here we need some further information on the uniformity of the constants. 
\begin{theorem}\label{P26}
Let  $T\in FIO(\chi)$. Then $T$ extends to a bounded
operator on $M^p(\rd)$, $1\leq p\leq\infty$ (and in particular on
$L^2(\rd) = M^2(\rd )$). Moreover, for every $m>2d$ there exists a constant $C>0$ depending only on m and the dimension $d$ such that
\[
\|T\|_{M^p\to M^p}\leq C\|T\|_{m,\chi}.
\]
\end{theorem}
\begin{proof}

By the inversion formula for the STFT in \eqref{inversione} we have, if $\|g\|_{L^2}=(2\pi)^{-d/2}$, 
 \[
 T= V_g^\ast V_g T V_g^\ast V_g.
 \]
Since $V_g:M^p(\rd)\to L^p(\rdd)$ and $V_g^\ast:L^p(\rdd)\to M^p(\rd)$, it suffices to prove that the operator $V_g T V_g^\ast$ is bounded on $L^p(\rdd)$. Now, this in an integral operator in $\rdd$ with integral kernel
 \[
 K(w,z)=\langle T\pi (z)g, \pi (w)g\rangle.
 \] 
Since $T\in FIO(\chi)$ we have, for every $m\in\bN$,
 \begin{align*}
 |V_g T V_g^\ast F(w)|&\leq \|T\|_{m,\chi}\int_{\rdd} \langle w-\chi(z)\rangle^{-m}
 |F(z)|\,dz.
 \end{align*}
 If we take $m>2d$, the desired continuity on $L^p(\rdd)$ follows at once from Schur's test and the fact that the Jacobian determinant of $\chi$ is $=1$, $\chi$ being symplectic.  
 \end{proof}

\begin{theorem}\label{prod1}
If $T^{(i)}\in FIO(\chi_i)$, $i=1,2$, then the
composition  $T^{(1)}T^{(2)}$ is in $FIO(\chi_1\circ \chi_2)$. Moreover for every $m>2d$ there exists a constant $C>0$ depending only on m, the dimension $d$, and upper bounds for the first derivatives of $\chi_1$ such that
\[
\|T^{(1)}T^{(2)}\|_{m,\chi_1\circ\chi_2}\leq C\|T^{(1)}\|_{m,\chi_1}\|T^{(2)}\|_{m,\chi_2}
\]
\end{theorem}
\begin{proof}
We have to estimate the integral kernel of $V_gT^{(1)}T^{(2)}V_g^\ast$. We write
 \[
  V_gT^{(1)}T^{(2)}V_g^\ast= V_g T^{(1)} V_g^\ast V_g T^{(2)}V_g^\ast,
  \]
 so that for $m>2d$
 \begin{align*}
 |\langle T^{(1)}&T^{(2)}\pi(z)g,\pi(w)g\rangle |\leq\|T^{(1)}\|_{m,\chi_1}\|T^{(2)}\|_{m,\chi_2}\int_\rdd \langle w-\chi_1(\zeta)\rangle^{-m}\langle\zeta-\chi_2(z)\rangle^{-m} d\zeta\\
 &\leq C \|T^{(1)}\|_{m,\chi_1}\|T^{(2)}\|_{m,\chi_2}\int_\rdd \langle w-\chi_1(\zeta)\rangle^{-m}\langle\chi_1(\zeta)-\chi_1\circ\chi_2(z)\rangle^{-m} d\zeta.
 \end{align*}
 The change of variable $\tilde{\zeta}=\chi_1(\zeta)-\chi_1\circ\chi_2(z)$ and the convolution property 
 \[
 \int_{\rdd} \langle w-\zeta\rangle^{-m}\langle \zeta\rangle^{-m}\,  d\zeta\leq C_{m,d}\langle w \rangle^{-m}
 \]
 (valid for $m>2d$, see e.g. \cite[Formula (11.5)]{book}) then give the desired estimate.
 \end{proof}
 The reader may want to compare the simplicity of this result with the similar composition formula obtained in \cite[Theorem A.2]{fujiwara2} for oscillatory integral operators, whose proof represented really a tour de force and moreover was limited to short time.\par 
We now show that a familiar class of Fourier integral operators is in fact contained in $FIO(\chi)$ for a suitable $\chi$. \par
A real phase function $\Phi$ on $\rdd$ will be called \emph{tame} if the
following three properties are satisfied:\par\medskip \noindent
{\it B1}. $\Phi\in \cC^{\infty}(\rdd)$;\par\medskip\noindent
{\it B2}. We have
\begin{equation}\label{phasedecay}
|\partial_x^\a\partial^\beta_\eta \Phi(x,\eta)|\leq C_{\a,\beta},\quad |\a|+|\beta|\geq 2,\ x,\eta\in\rd;\end{equation}
{\it B3}. There exists $\tilde{\delta}>0$ such that
\begin{equation}\label{detcond}
  |\det\,\partial^2_{x,\eta} \Phi(x,\o)|\geq \tilde{\delta},\quad x,\eta\in\rd.
\end{equation}
\par
If we set
\begin{equation}\label{cantra} \left\{
               \begin{array}{l}
               y=\nabla_{\eta}\Phi(x,\eta)
               \\
              \xi=\nabla_{x}\Phi(x,\eta), \rule{0mm}{0.55cm}
               \end{array}
               \right.
\end{equation}
 we can solve with respect to $(x,\xi)$ by the global inverse
function theorem and  obtain a mapping
$\chi$ defined by $(x,\xi)=\chi(y,\o)$ and satisfying {\it A1, A2} above  as well as the following property:\par\medskip\noindent
{\it A3} There exists $\tilde{\delta}>0$ such that, \begin{equation}\label{detcond2}  |\det\,\frac{\partial x}{\partial y}(y,\eta)|\geq \tilde{\delta},\quad y,\eta\in\rd.
 \end{equation}
Conversely,  to every transformation $\chi$ satisfying {\it A1},
{\it A2} and {\it A3} 
corresponds a tame phase $\Phi$, uniquely
determined up to a constant (see e.g.\ \cite{cgnr}). \par
 The following result was proved in \cite[Theorem 3.3]{cnr1}. 
\begin{theorem}\label{teo3.5}
Let $\Phi(x,\eta)$ be a tame phase, and  let $\chi$ be the corresponding canonical transformation. Let $a(x,\eta)$ be a function in $S^0_{0,0}$. The Fourier integral operator  
\[
T f(x)=(2\pi)^{-d}\int_{\rd} e^{i\Phi(x,\eta)}a(x,\eta)\widehat{f}(\eta)\, d\eta
\]
belongs to $FIO(\chi)$. Moreover for every $m\in\bN$ there exists $m'\in\bN$ such that
\[
\|T\|_{m,\chi}\leq C\|a\|_{m'} 
\]
for some constant $C$ depending only on $m$, the dimension $d$, upper bounds for a certain number of the derivatives of $\Phi$ in \eqref{phasedecay} and the lower bound constant $\tilde{\delta}$ in \eqref{detcond}.
\end{theorem}

\subsection{Semiclassical phase space methods}
Here we rephrase in the semiclassical regime the concepts appeared above in this section. Basically we have to rescale function spaces and operators with the dilation operators (unitary on $L^2(\rd)$)
\[
D_{\hbar^{-1/2}} f(x)=\hbar^{d/4} f(\hbar^{1/2} x),\quad D_{\hbar^{1/2}} f(x)=\hbar^{-d/4} f(\hbar^{-1/2}x).
\]
We already defined in Introduction the spaces $\hlp$, which are a rescaled version of the Sobolev spaces $L^p_k$. For $1\leq p\leq\infty$, consider also the modulation spaces
\[
M^p_\hbar=\{f\in\cS'(\rd):\ \|f\|_{M^p_\hbar}=\| D_{\hbar^{-1/2}} f\|_{M^p}<+\infty\}.
\]
We have the following embeddings.
\begin{theorem}\label{immersionisemi}
Let $1<p<\infty$ and $k=k_p=2d|1/2-1/p|$. There exists a constant $C>0$ such that 
\begin{equation}\label{imm1}
\|f\|_{L^p}\leq C \hbar^{d(1/p-1/2)/2}\|f\|_{M^p_\hbar},\quad \|f\|_{M^p_\hbar}\leq C\hbar^{d(1/2-1/p)/2}\|f\|_{\tilde{L}^p_{k}}\quad 1<p\leq 2
\end{equation}
as well as
\begin{equation}\label{imm2}
\|f\|_{M^p_\hbar}\leq C \hbar^{d(1/2-1/p)/2}\|f\|_{L^p},\quad \|f\|_{\tilde{L}^p_{-k}}\leq C\hbar^{d(1/p-1/2)/2}\|f\|_{M^p_\hbar}\quad 2\leq p<\infty.
\end{equation}
\end{theorem}
\begin{proof}
By \eqref{immersioni} we have, for $1<p\leq 2$ and $k=2d|1/p-1/2|$,
\[
\|f\|_{L^p}\leq C \|f\|_{M^p},\quad \|f\|_{M^p}\leq C\|f\|_{L^p_k}.
\]
Now we replace $f$ by $D_{\hbar^{-1/2}} f$ and we obtain \eqref{imm1}, because
\[
\|D_{\hbar^{-1/2}} f\|_{L^p}=\hbar^{d(1/2-1/p)/2}\|f\|_{L^p}\quad{\rm and}\quad 
\|D_{\hbar^{-1/2}} f\|_{L^p_k}=\hbar^{d(1/2-1/p)/2}\|f\|_{\tilde{L}^p_k}.
\]
Similarly one deduces \eqref{imm2}.
\end{proof}
\begin{definition}
Let $\chi$ be a tame canonical transformation. We denote by $FIO_\hbar(\chi)$ the space of linear operators $T:\cS(\rd)\to\cS'(\rd)$ such that $D_{\hbar^{-1/2}} T D_{\hbar^{1/2}}\in FIO(\chi)$, endowed with the seminorms
\[
\|T\|^\hbar_{m,\chi}=\| D_{\hbar^{-1/2}} T D_{\hbar^{1/2}}\|_{m,\chi}.
\]
\end{definition}
From Theorems \ref{P26} we obtain at once the following result.
\begin{theorem}\label{prodsemi0}
Let  $T\in FIO_\hbar(\chi)$. Then $T$ extends to a bounded
operator on $M^p_\hbar(\rd)$, $1\leq p\leq\infty$. Moreover, for every $m>2d$ there exists a constant $C>0$ depending only on m and the dimension $d$ such that
\[
\|T\|_{M^p_\hbar\to M^p_\hbar}\leq C\|T\|^\hbar_{m,\chi}.
\]
\end{theorem}
\begin{corollary}\label{corollario}
Let  $T\in FIO_\hbar(\chi)$, $1<p<\infty$ and $k=k_p=2d|1/p-1/2|$. Then $T$ extends to a bounded operator $T:\hlp\to L^p$ if $1<p\leq 2$ and $T:\tilde{L}^p\to \tilde{L}^p_{-k}$ for $2\leq p<\infty$. Moreover for every $m>2d$ there exists a constant $C>0$ depending on $m,d,p$ such that
\[
\|T\|_{\hlp\to L^p}\leq C\|T\|^\hbar_{m,\chi} \quad{\rm for} \ 1<p\leq 2
\]
and 
\[
\|T\|_{L^p\to \tilde{L}^p_{-k}}\leq C\|T\|^\hbar_{m,\chi} \quad{\rm for} \ 2\leq p<\infty. 
\]
\end{corollary}
\begin{proof}
The results follow at once from Theorem \ref{prodsemi0} and Theorem \ref{immersionisemi}. For example, for $1<p\leq 2$ we have
\begin{align*}
\|Tf\|_{L^p}&\leq C_1 \hbar^{d(1/p-1/2)/2}\|Tf\|_{M^p_\hbar}\leq C_1C_2\hbar^{d(1/p-1/2)/2}\|T\|^\hbar_{m,\chi}\|f\|_{M^p_\hbar}\\
&\leq C_1C_2C_3\|T\|^\hbar_{m,\chi}\|f\|_{\hlp}.
\end{align*}
\end{proof}

The following result will be crucial in the sequel. 
\begin{theorem}\label{prodsemi}
If $T^{(i)}\in FIO_\hbar(\chi_i)$, $i=1,2,$ then the
composition  $T^{(1)}T^{(2)}$ is in $FIO_\hbar(\chi_1\circ \chi_2)$. Moreover for every $m>2d$ there exists a constant $C>0$ depending only on m, the dimension $d$, and upper bounds for the first derivatives of $\chi_1$ such that
\[
\|T^{(1)}T^{(2)}\|^\hbar_{m,\chi_1\circ\chi_2}\leq C\|T^{(1)}\|^\hbar_{m,\chi_1}\|T^{(2)}\|^\hbar_{m,\chi_2}
\]
\begin{proof}
We write 
\[
D_{\hbar^{-1/2}} T^{(1)}T^{(2)} D_{\hbar^{1/2}}=D_{\hbar^{-1/2}} T^{(1)} D_{\hbar^{1/2}} D_{\hbar^{-1/2}} T^{(2)} D_{\hbar^{1/2}}
\]
and apply Theorem \ref{prod1}. 
\end{proof}

\end{theorem}

\section{The Schr\"odinger flow: wave packet analysis}
In this section we prove that the oscillatory integral operators appearing in Section 2 belong to the class $FIO_\hbar(\chi^\hbar(t,s))$ for a suitable canonical transformation $\chi^\hbar$ defined below. Then we will use Theorem \ref{prodsemi} to treat their composition for large time. \par
We adopt the notation of Section 2. In particular the oscillatory integral operators are well defined for $|t-s|<\delta$.
\begin{proposition}\label{prop7}
Let $|t-s|<\delta$ and $a\in S^0_{0,0}$. Consider the operator
\[
Tf(x)=\frac{1}{(2\pi i (t-s) \hbar)^{d/2}} \int_{\rd} e^{i\hbar^{-1}S(t,s,x,y)} a(x,y)f(y)\, dy.
\]
Then $T$ can be written as a Fourier integral operators
\begin{equation}\label{fio1}
Tf(x)=(2\pi)^{-d} \int_{\rd} e^{i\hbar^{-1}\Phi(t,s,x,\hbar\eta)} b(\hbar,t,s,x,\hbar\eta)\widehat{f}(\eta)\, d\eta
\end{equation}
where 
\[
\Phi(t,s,x,\eta)=y\eta+S(t,s,x,y)
\]
is the generating function in the coordinates\footnote{Namely, here $y= y(t,s,x,\eta)$ is the unique solution to $-\eta=\partial S(t,s,x,y)/\partial y$.} $x,\eta$, and $a\in S^0_{0,0}$. Moreover for every $m\in\bN$ there exists $m'\in\bN$ such that
\[
\|b(\hbar,t,s,\cdot,\cdot)\|_m\leq C\|a\|_{m'}
\]
for some constant $C>0$, depending only on $m$, the dimension $d$, upper bounds for a certain number of the derivatives of $(t-s)S(t,s,x,y)$ in \eqref{pre1} and the lower bound constant $\tilde{\delta}$ in \eqref{pre2}.
\end{proposition}
\begin{proof}
Set $\tilde{S}(t,s,x,y)=(t-s)S(t,s,x,y)$ and $\tilde{\eta}=\hbar(t-s)\eta$. By the Plancherel theorem we can write $T$ in the form
\begin{multline*}
Tf(x)\\=\frac{1}{(2\pi i (t-s) \hbar)^{d/2}} \int_{\rd}\Big(\int_{\rd} \exp\Big({i(\hbar(t-s))^{-1}[y\tilde{\eta}+\tilde{S}(t,s,x,y)]}\Big)a(x,y)\,dy\Big)\widehat{f}(\eta)\,d\eta.
\end{multline*}
 We now apply the stationary phase principle in the form proved in \cite[page 320 and Lemma 3.2]{asada}; the assumptions are satisfied because for $|t-s|<\delta$ we have, by \eqref{pre1} and \eqref{pre2},
\[
|\partial^\alpha_x\partial^\beta_y \tilde{S}(t,s,x,y)|\leq C_{\alpha,\beta}\quad |\alpha|+|\beta|\geq 2,
\]
and 
\[
\Big|{\rm det}\,\Big(\frac{\partial^2 \tilde{S}(t,s,x,y)}{\partial y^2}\Big)\Big|\geq \tilde{\delta}>0.
\]
Moreover, since $y=y(t,s,x,\eta)$ is the unique solution of $-\eta=\partial S(t,s,x,y)/\partial y$, the function  $y=y(t,s,x,\tilde{\eta}/(t-s))$ will be the unique solution to $-\tilde{\eta}=\partial\tilde{S}(t,s,x,y)/\partial y$, and we obtain
\begin{multline*}
\frac{1}{(2\pi i (t-s) \hbar)^{d/2}}\int_{\rd} \exp\Big({i(\hbar(t-s))^{-1}[y\tilde{\eta}+\tilde{S}(t,s,x,y)]}\Big)a(x,y)\,dy\\
=\exp\Big(i(\hbar(t-s))^{-1}[y(t,s,x,\tilde{\eta}/(t-s))\tilde{\eta}+\tilde{S}(t,s,x,y(t,s,x,\tilde{\eta}/(t-s)))]\Big)\\
\times
b(\hbar,t,s,x,\tilde{\eta}))
\end{multline*}
for some amplitude $b(\hbar,t,s,\cdot,\cdot)$ belonging to some bounded subset of $S^0_{0,0}$, when $|t-s|<\delta$, $0<\hbar\leq 1$. Again, we can write the last expression as 
\[
\exp\Big(i\hbar^{-1}[y(t,s,x,\hbar \eta)\hbar\eta+S(t,s,x,y(t,s,x,\hbar\eta))]\Big) b(\hbar,t,s,x,(t-s)\hbar\eta),
\]
which gives the desired expression for $Tf$.\par
 The estimates of the seminorms of $b$ in terms of those of $a$ follow from the proof of the stationary phase principle.
\end{proof}

\begin{corollary}\label{prop8}
With the notation of Proposition \ref{prop7}, we have $T\in FIO_\hbar(\chi^\hbar(t,s))$, with 
\[
\chi^\hbar(t,s)(y,\eta)=\hbar^{-1/2}\chi(t,s,\hbar^{1/2}y,\hbar^{1/2}\eta).
\]
Moreover, for every  $m\in\bN$ there exists $m'\in\bN$ such that, for $|t-s|<\delta$,
\[
\|T\|^\hbar_{m,\chi^\hbar(t,s)}\leq C\|a\|_{m'}
\]
for some constant $C>0$ independent of $a$ and $\hbar$, $t,s$ ($0<\hbar\leq1$, $t-s<\delta$).
\end{corollary}
\begin{proof}
We have to prove that $D_{\hbar^{-1/2}} T D_{\hbar^{1/2}}\in FIO(\chi^\hbar(t,s))$. Using \eqref{fio1} and a scaling argument we can write
\[
D_{\hbar^{-1/2}} T D_{\hbar^{1/2}}f(x)=(2\pi)^{-d} \int_{\rd} e^{i\hbar^{-1}\Phi(\hbar^{1/2}x,\hbar^{1/2}\eta)} b(\hbar,t,s,\hbar^{1/2}x,\hbar^{1/2}\eta)\widehat{f}(\eta)\, d\eta.
\] 
Now, the phase $\hbar^{-1}\Phi(t,s,\hbar^{1/2}x,\hbar^{1/2}\eta)$ generates the canonical transformation $\chi^\hbar(t,s)$ as in the statement. Moreover $\hbar^{-1}\Phi(t,s,\hbar^{1/2}x,\hbar^{1/2}\eta)$ is tame {\it uniformly} with respect to $\hbar$ and $t,s$, for $0<\hbar\leq 1$ and $|t-s|<\delta$, in the sense that the required bounds hold with constants independent of these parameters, or equivalently $\chi^\hbar(t,s)$ satisfies the properties {\it A1,A2,A3} in Section 3.2 above, with uniform bounds. This is clear, because $\chi(t,s)$ satisfies {\it A1,A2,A3} uniformly with respect to $t,s$ for $|t-s|<\delta$ (possibly for a smaller value of $\delta$): {\it A1,A2} follow from \eqref{stimeflusso}, whereas {\it A3} holds (even in the presence of a magnetic field as in Section 6 below) by \cite[Proposition 2.3']{yajima} applied with $\alpha=\beta=0$.\par
Finally, the symbol $b(\hbar,t,s,\hbar^{1/2}x,\hbar^{1/2}\eta)$ has seminorms in $S^0_{0,0}$ dominated by those of $b(\hbar,t,s,\cdot,\cdot)$. Hence the desired result follows from Theorem \ref{teo3.5}.
\end{proof}
\begin{remark}\label{oss1} With the notation of the previous corollary, notice that
for any $s,\tau,t\in\R$ we have 
\[
\chi^\hbar(t,s)=\chi^\hbar(t,\tau)\circ\chi^\hbar(\tau,s).
\]
Indeed, this follows at once from the case $\hbar=1$ and the definition of $\chi^\hbar$.
\end{remark}
\section{Proof of the main result (Theorem \ref{mainteo})}
 Let us first prove \eqref{uno} and \eqref{due}. We know from Section 2 that the propagator $U(t,s)$ is an oscillatory integral operator of the form \eqref{cinque0} for $|t-s|<\delta$, and therefore by Corollary \ref{prop8} we have $U(t,s)\in FIO_\hbar(\chi^\hbar(t,s))$, with seminorms $\|U(t,s)\|^\hbar_{m,\chi^\hbar(t,s)}$, $m\in\bN$, uniformly bounded with respect to $\hbar,t,s$ for $|t-s|<\delta$. Using the evolution properties of the propagator, Remark \ref{oss1} and Theorem \ref{prodsemi} we see that, for any $T>0$ we have  $U(t,s)\in FIO_\hbar(\chi^\hbar(t,s))$ with seminorms
 \begin{equation}\label{chiave0}
 \|U(t,s)\|^\hbar_{m,\chi^\hbar(t,s)}\leq C_0
 \end{equation}
 for a constant $C_0$ independent of $\hbar,t,s$, for $|t-s|\leq T$ (but depending on $T$); here we are using the fact that constant $C$ in Theorem \ref{prodsemi} can be chosen independent of such parameters, being $\chi^\hbar(t,s)$ uniformly tame. We then deduce \eqref{uno} and \eqref{due} from Corollary \ref{corollario}. \par
We now prove \eqref{tre} and \eqref{quattro}. By arguing as above it suffices to prove that
\begin{equation}\label{daver}
\|E^{(N)}(\Omega,t,s)-U(t,s)\|^\hbar_{m,\chi^\hbar(t,s)}
\leq C(T)\hbar^N \omega(\Omega)^{N+1}|t-s|
\end{equation}
for $|t-s|\leq T$ and for some $m>2d$.\par
 Now, by Corollary \ref{prop8} and \eqref{stimeresto} the operator $G^{(N)}(t,s)$ defined in \eqref{defgn} belongs to $FIO_\hbar(\chi^\hbar(t,s))$, with seminorms 
\begin{equation}\label{34bis}
\|G^{(N)}(t,s)\|^\hbar_{m,\chi^\hbar(t,s)}\leq C_m\hbar^{N+1}|t-s|^{N+1}
\end{equation}
for a constant $C_m$ independent of $\hbar,t,s$, for $|t-s|<\delta$.
On the other hand we have 
\[
R^{(N)}(t,s)f:=E^{(N)}(t,s)f-U(t,s)f=-i\hbar^{-1}\int_s^t U(t,\tau)G^{(N)}(\tau,s)f\,d\tau
\]
so that by \eqref{chiave0}, \eqref{34bis} and Theorem \ref{prodsemi} we obtain
\begin{equation}\label{chiave}
\|R^{(N)}(t,s)\|^\hbar_{m,\chi^\hbar(t,s)}\leq C_m\hbar^{N}|t-s|^{N+2}.
\end{equation}
We then can write
\[
E^{(N)}(\Omega,t,s)-U(t,s)=\big(U(t,t_{L-1})+R^{(N)}(t,t_{L-1})\big)\ldots \big(U(t_1,s)+R^{(N)}(t_1,s)\big)-U(t,s)
\]
and argue as in \cite[Lemma 3.2]{fujiwara2}, applying the composition estimate in Theorem \ref{prodsemi}. We sketch the argument for the benefit of the reader. \par
One expands the above product and obtains a sum of ordered products of operators, where each product has the following structure: {\it from right to left} we have, say, $q_1$ factors of type $U$, $p_1$ factors of type $R^{(N)}$, $q_2$ factors of type $U$, $p_2$ factors of type $R^{(N)}$, etc., up to $q_k$ factors of type $U$, $p_k$ factors of type $R^{(N)}$, to finish with $q_{k+1}$ factors of type $U$. Here $p_1,\ldots,p_k,q_1,\ldots q_k,q_{k+1}$ are non negative integers whose sum is $L$, with $p_j>0$ and we can of course group together the consecutive factors of type $U$, using the evolution property of the propagator. Now, we estimate the seminorm $\|\cdot\|^\hbar_{m,\chi^\hbar(t,s)}$ of such an ordered product, for some fixed $m>2d$, using Theorem \ref{prodsemi} combined with the known estimates for each factor, namely \eqref{chiave0} and \eqref{chiave}, which we rephrase here as
\begin{equation}\label{chiave2}
\|R^{(N)}(t_j,t_{j-1})\|^\hbar_{m,\chi^\hbar(t,s)}\leq \tilde{C}_0\hbar^{N}|t_j-t_{j-1}|^{N+2}.
\end{equation}
If $C$ denotes the constant in Theorem \ref{prodsemi} we then dominate the seminorm $\|\cdot\|^\hbar_{m,\chi^\hbar(t,s)}$ of the above ordered product by
\[
C^{p_1+\ldots+p_k+k}C_0^{k+1} \prod_{j=1}^k \prod_{i=1}^{p_j} \tilde{C}_0\hbar^{N}|t_{J_j+i}-t_{J_{j}+i-1}|^{N+2}
\]
where $J_j=p_1+\ldots+p_{j-1}+q_1+\ldots +q_j$ for $j\geq2$ and $J_1=q_1$. Now, it was proved in the last part of the proof of \cite[Lemma 3.2]{fujiwara2} that the sum over $p_1,\ldots,p_k,q_1,\ldots,q_{k+1}$ of these terms is in turn dominated by 
\[
C(T)\hbar^N \omega(\Omega)^{N+1}|t-s|
\]
for $|t-s|\leq T$. This gives \eqref{daver} and concludes the proof.

\section{The case of magnetic fields}
In this section we extend some of the previous results in the presence of a magnetic field. Hence, consider the Schr\"odinger equation
\begin{equation}\label{equazione}
i\hbar \partial_t u=\frac{1}{2}\big(-i\hbar\nabla-A(t,x)\big)^2 u+V(t,x)u
\end{equation}
where $V(t,x)$ and $A(t,x)=(A_1(t,x),\ldots, A_d(t,x))$ are electric scalar and magnetic vector potential of the field, $t\in \R$, $x\in\rd$.\par
 Assume the following hypothesis (cf. \cite{fujiwara5,tsuchida,yajima}). \par\medskip
{\bf Assumption (B)} {\it \begin{itemize}
\item[a)] For $j=1,\ldots,d$, $A_j(t,x)$ is a real function of $(t,x)\in \R\times\rd$ and $\partial^\alpha_x A_j(t,x)$ is $C^1$ in $(t,x)\in \R\times\rd$, for every $\alpha\in\bN^d$. Moreover there exists $\epsilon>0$ such that
\[
|\partial^\alpha_x B(t,x)|\leq C_\alpha(1+|x|)^{-1-\epsilon},\quad|\alpha|\geq 1
\]
\[
|\partial^\alpha_x A(t,x)|+|\partial^\alpha_x\partial_t A(t,x)|\leq C_\alpha,\quad|\alpha|\geq 1, 
\]
for $(t,x)\in \R\times\rd$, where $B(t,x)$ is the magnetic field, i.e.\ the skew-symmetric matrix with entries $B_{j,k}(t,x)=(\partial A_k/\partial x_j-\partial A_j/\partial x_k)(t,x)$.\par\smallskip
\item[b)] $V(t,x)$ is a real function of $(t,x)\in \R\times\rd$ with $\partial^\alpha_x V(t,x)$ continuous in $(t,x)\in \R\times \rd$, for every $\alpha\in\bN^d$ and satisfying
\[
|\partial^\alpha_x V(t,x)|\leq C_\alpha,\quad|\alpha|\geq 2,\ (t,x)\in \R\times \rd.
\]
\end{itemize}}
It was proved in \cite[Sections 2,3]{yajima} that if $V=0$, under Assumption (B), all the short time results about the propagator, as well as the construction of the operators $E^{(N)}(t,s)$, $G^{(N)}(t,s)$ summarized in Section 2 continue to hold if $N\geq 1$. To be precise, the same formulas for $E^{(N)}(t,s)$, $G^{(N)}(t,s)$ and $E^{(N)}(\Omega,t,s)$ hold, where now the amplitudes $a_k$ satisfy the transport equations
\begin{multline}\label{trasporto2}
\frac{\partial a_k}{\partial t}+\sum_{j=1}^d\Big(\frac{\partial S}{\partial x_j}(t,s,x,y)-A_j(t,x)\Big)\frac{\partial a_k}{\partial x_j}\\
+\frac{1}{2}\Big(\Delta_x S(t,s,x,y)-\frac{d}{t-s}-{\rm div}_x\, A(t,x)\Big)a_k=\frac{1}{2}\Delta_x a_{k-1},
\end{multline}
with $a_{0}(t,s,x,y)\equiv0$ and the initial conditions $a_1(s,s,x,y)=1$ and $a_k(s,s,x,y)=0$ for $k=1,2,\ldots$ (these transport equations hold even when the electric potential is present\footnote{Indeed, since $S$ satisfies the Hamilton-Jacobi equation $\partial_t S+\frac{1}{2}(\nabla_x S-A)^2+V=0$, we have 
\begin{multline*}
\Big(i\hbar\partial_t-\frac{1}{2}(-i\hbar\nabla_x-A)^2-V(t,x)\Big)\frac{1}{(2\pi i \hbar(t-s))^{d/2}}\int_{\rd} e^{i\hbar^{-1}S(t,s,x,y)} a(x,y)\, dy\\
=\frac{1}{(2\pi i \hbar(t-s))^{d/2}}\int_{\rd} e^{i\hbar^{-1}S(t,s,x,y)}i\hbar\Big[\partial_t a+(\nabla_x S-A)\cdot\nabla_x a+\frac{1}{2}(\Delta_x S-\frac{d}{t-s}-{\rm div}_x\,A)-\frac{i\hbar}{2}\Delta_x a \Big]\, dy.
\end{multline*}
}). 
The key estimates \eqref{quattro0}--\eqref{cinque1} are valid too, for $N\geq 1$. \par
We have therefore the following result.
\begin{theorem}\label{magnetic}
Under the above Assumption (B), the conclusions of Theorem \ref{mainteo} still hold, at least for $N\geq 1$.
\end{theorem}
\begin{proof}
As we have already observed, the results in Section 2 keep valid under the Assumption (B), if $V=0$ and $N\geq 1$, therefore the whole subsequent analysis applies, giving the desired result if $V=0$. \par
Consider now the case when a potential $V(t,x)$ is present. Let us prove just part b) of Theorem \ref{mainteo} (the first part is similar and easier). We apply the gauge transformation
\[
\mathcal{G}(t)u(t,x)=\exp(i\hbar^{-1}G(t,x))u(t,x)
\]
which maps the solution of the equation with potentials $A,V$ into the solution of the equation with potentials $A+\nabla_xG$, $V-\partial_t G$, cf.\ \cite{yajima}. More precisely, if we denote by $U(t,s)$ and $U'(t,s)$ the propagators, with the prime denoting that of the new equation, we have 
\[
U(t,s)=\mathcal{G}(t)^{-1}U'(t,s)\mathcal{G}(s).
\]
In particular if we choose $G(t,x)=\int_0^t V(\tau,x)\, d\tau$, in the new equation the electric potential is absent and and magnetic potential is $A'=A+\int_0^t\partial_x V(\tau,x)d\tau$, still satisfying the above Assumption (B). Since the Lagrangian function $L(x,v,t)=\frac{1}{2}|v|^2+A(t,x)\cdot v-V(t,x)$ changes to $L'=L+dG(t,x(t))/dt$, the flow remains the same and the new generating function is 
\[
S'(t,s,x,y)=S(t,s,x,y)+\int_0^t V(\tau,x)\, d\tau-\int_0^s V(\tau,y)\, d\tau
\]
where $S(t,s,x,y)$ is the generating function for the original equation. The amplitudes $a_k$ for the new equation are the same as those for the original one, because the transport equations  \eqref{trasporto2} do not change under the substitution $A\to A+\int_0^t\nabla_x V(\tau,x)\,d\tau$, $S\to S'$, $V\to 0$. As a consequence we have
\[
E^{(N)}(t,s)=\mathcal{G}^{-1}(t)E^{(N)'}(t,s)\mathcal{G}(s)
\]
and therefore 
\[
E^{(N)}(\Omega,t,s)=\mathcal{G}^{-1}(t)E^{(N)'}(\Omega,t,s)\mathcal{G}(s)
\]
where the the primes denote the corresponding operators for the new equation. As a consequence
\[
E^{(N)}(\Omega,t,s)-U(t,s)=\mathcal{G}^{-1}(t)\Big(E^{(N)'}(t,s)-U'(t,s)\Big)\mathcal{G}(s).
\]
Now, we know that \eqref{daver} holds for the difference $E^{(N)'}(\Omega,t,s)-U'(t,s)$ and we want to prove that the same is true for $E^{(N)}(\Omega,t,s)-U(t,s)$.  Using Theorem \ref{prodsemi}, it is then sufficient to prove that $\mathcal{G}(t)\in FIO_\hbar(\tilde{\chi}^\hbar)$ for some tame canonical transformation $\tilde{\chi}$. This follows by Corollary \ref{prop8}, because $\mathcal{G}(t)$ can be written in the form \eqref{fio1} with $\Phi(t,x,\eta)=x\eta+\int_0^t V(\tau,x)\,d\tau$, $b\equiv1$, and 
\[
\tilde{\chi}(y,\eta)=\Big(y, \eta+\int_0^t \nabla_x V(\tau,x)\, d\tau\Big).
\]
\end{proof}
\begin{remark}
We do not know whether the conclusion of Theorem \ref{mainteo} holds for $N=0$ in the presence of a magnetic field as above (in fact, in \cite[Theorem 5]{yajima} the case $N=0$ is excluded as well). Indeed, the approach in \cite{fujiwara2} requires very precise estimates on the generating function $S(t,s,x,y)$, which seem to be quite tricky when a magnetic field is present. 
\end{remark}

\section{Concluding remarks}
\subsection{Sharpness of the results}
In general we cannot avoid the dichotomy $p>2$, $p<2$ in Theorem \ref{mainteo}. Consider, for example, the case of the harmonic oscillator, 
\begin{equation}\label{arosc}
i\partial_t u=-\frac{1}{2}\Delta u+\frac{1}{2}|x|^2 u
\end{equation}
(we take $\hbar=1$ for simplicity, so that $\hlp=L^p_k$). By the Mehler formula (see e.g. \cite[formula (5a)]{rod}), the propagator $U(t,0)$ at $t=\pi/2$ is (up to constant factor) just the Fourier transform $f\mapsto \widehat{f}$. Now, we have the following result.
\begin{proposition}
Let $1<p<\infty$, $k_1,k_2\in\R$. The Fourier transform maps $L^p_{k_1}\to L^p_{k_2}$ continuously if and only if 
\[
k_1\geq 2d(1/p-1/2)\ \ {\rm and}\ \ k_2\leq 0,\quad {\rm if}\ 1<p\leq 2
\]
\[
k_1\geq 0\ \ {\rm and}\ \ k_2\leq -2d(1/2-1/p),\quad {\rm if}\ 2\leq p< \infty.
\]
\end{proposition}
\begin{proof}
{\it Sufficient conditions}. 
The desired estimates follow by the inclusion relations of Sobolev spaces and duality if we prove the continuity of the Fourier transform $L^p_{k}\to L^p$ with $k=2d(1/p-1/2)$, for $1<p\leq 2$. This is equivalent to the estimate
\[
\| (1+|\xi|^2)^{-k/2}\widehat{f}(\xi)\|_{L^p}\leq C\|f\|_{L^p},
\]
which in turn is a consequence of the Hardy-Littlewood-Paley inequality \cite[Theorem 1.4.1]{bl} (that is the same estimate with $(1+|\xi|^2)^{-k/2}$ replaced by $|\xi|^{-k}$ in the left-hand side).\par
Of course the result also follows from Theorem \ref{mainteo} applied to the equation \eqref{arosc}, with $s=0$, $t=\pi/2$ ($\hbar=1$).\par
{\it Necessary conditions}. Let us first prove the condition on $k_2$ when $2\leq p<\infty$. By duality this will give the condition on $k_1$ for $1<p\leq 2$ as well.\par
Consider the space 
\[
\mathcal{A}=\{f\in\cS(\rd):\  \widehat{f}(\xi)=0\ {\rm for}\ |\xi|\geq 1\}.
\]
By Bernstein inequalities (see e.g. \cite[Proposition 5.3]{wolff}) there exists a constant $C>0$ such that
\[
\|f\|_{L^p_{k_1}}\leq C\|f\|_{L^p} \quad \forall f\in\mathcal{A},
\]
even for $k_1< 0$ (for $k_1\geq0$ this is trivially true for every $f\in\cS(\rd)$). \par
Hence, if the Fourier transform is bounded  $L^p_{k_1}\to L^p_{k_2}$, we have 
\begin{equation}\label{app1}
\|\widehat{f}\|_{L^p_{k_2}}\leq C\|f\|_{L^p}\quad \forall f\in\mathcal{A}.
\end{equation}
Now, suppose by contradiction that $k_2>-2d(1/2-1/p)$. Fix $f\in\mathcal{A}\setminus\{0\}$ and test this estimate on $f(x/\lambda)$, with $\lambda\geq 1$, so that $f(x/\lambda)$ belongs to $\mathcal{A}$ too.
We have 
\begin{equation}\label{41bis}
\|(1-\Delta)^{k_2/2}\widehat{f(\cdot /\lambda)}\|_{L^p}\leq C\|f(\cdot /\lambda)\|_{L^p}
\end{equation}
where
\[
(1-\Delta)^{k_2/2}\widehat{f(\cdot /\lambda)}=\lambda^{d+k_2}[(\lambda^{-2}-\Delta)^{k_2/2}\widehat{f}](\lambda\cdot).
\]
By dominated convergence theorem ($|x|^{k_2}$ is locally integrable, because $k_2>-d$) we have $[(\lambda^{-2}-\Delta)^{k_2/2}\widehat{f}](x)\to |D|^{k_2}\widehat{f}(x)$ for every $x\in\rd$, and therefore by the Fatou theorem,
\[
0\not=\||D|^{k_2}\widehat{f}\|_{L^p}\leq\liminf_{\lambda\to+\infty}\|(\lambda^{-2}-\Delta)^{k_2/2}\widehat{f}\|_{L^p}.
\]
Hence, letting $\lambda\to+\infty$ in \eqref{41bis} we obtain $k_2\leq -2d(1/2-1/p)$, which is a contradiction.\par
Let us now prove that $k_2\leq 0$ if $1<p\leq2$. This will give the condition on $k_1$ for $2\leq p<\infty$ as well.\par
Fix $f\in\cS(\rd)\setminus\{0\}$, and let $f_\lambda(x)=f(x_1-\lambda,x_2,\ldots,x_d)$, $\lambda>0$. Suppose the following estimate holds:
\begin{equation}\label{app2}
\|\widehat{f_\lambda}\|_{L^p_{k_2}}\leq C\|f_\lambda\|_{L^p_{k_1}}= C\|f\|_{L^p}.
\end{equation}
We have 
\begin{align*}
(1&-\Delta)^{k_2/2}\widehat{f_\lambda}(\xi)=(1-\Delta)^{k_2/2}[ e^{-i\lambda\xi_1}\widehat{f}(\xi)]\\
&=\int_{\rd} e^{ix\xi}(1+|x|^2)^{k_2/2}{f}(-(x_1+\lambda),-x_2,\ldots,-x_d)\,dx\\
&=\lambda^{k_2}e^{-i\lambda \xi_1}\int_{\rd} e^{-iy\xi}(\lambda^{-2}+|1+\lambda^{-1}y_1|^2+\lambda^{-2}|y_2|^2+\ldots+\lambda^{-2}|y_d|^2)^{k_2/2}{f}(y)\, dy.
\end{align*}
As $\lambda\to+\infty$ the last integral  converges to $\widehat{f}(\xi)$, for every $\xi\in \rd$, by the dominated convergence theorem. One can then conclude as above, by letting $\lambda\to+\infty$ in \eqref{app2} and using Fatou theorem. 

\end{proof}
\subsection{The case of rough potentials}
Here we present an extension of the first part of Theorem \ref{mainteo} to a class of potentials in the so-called Sj\"ostrand class
\[
M^{\infty,1}(\rd)=\{f\in\cS'(\rd): \|f\|_{M^{\infty,1}(\rd)}:=\int_{\rd}\|V_g f(\cdot,\xi)\|_{L^\infty}\, d\xi<+\infty\}
\]
where, as usual, $g\in\cS(\rd)\setminus\{0\}$.\par
Functions in this space are bounded in $\rd$ and locally have the mild regularity of a function whose Fourier transform is in $L^1$. Hence, it is not required the existence of any derivative. We refer to the books \cite{book,deGosson} for a detailed analysis of this function space and its role in Time-frequency Analysis and Mathematical Physics. Here we just need the notion of narrow continuity for a one-parameter family $t\mapsto f_t$ of functions in this space. In fact, it is quite a weak notion of continuity.\par
\begin{definition} Let $I\subset\R$ be an interval. We say that a map $I\ni t\mapsto f_t\in M^{\infty,1}(\rd)$ is continuous for the narrow convergence if it is continuous in $\cS'(\rd)$ (weakly) and if there exists a function $H\in L^1(\rd)$ such that $\sup_{x\in \rd}|V_g f_t (x,\xi)|\leq H(\xi)$ for every $t\in I$ and almost every $\xi\in\rd$. 
\end{definition}
Consider now the equation (with $\hbar=1$, for simplicity)
\[
i\partial_t u=(-\Delta)^{\kappa/2}u+V_2(t,x) u+V_1(t,x) u+V_0(t,x) u\\
\]
with $0<\kappa\leq 2$, $t\in\R$, $x\in\rd$. Suppose 
\[
\partial^\alpha_x V_j(t,\cdot)\in M^{\infty,1}(\rd)\ {\rm for}\ |\alpha|=j,\ j=0,1,2,
\]
with $V_2$ and $V_1$ real-valued. Assume moreover that the map $\R\ni t\mapsto \partial^\alpha_x V_j(t,\cdot)\in M^{\infty,1}(\rd)$ for $|\alpha|=j$, $j=0,1,2$, are continuous in the sense of the narrow convergence.\par
Under this assumption the propagator $U(t,s)$ was constructed in \cite{cnr10} (see also \cite{cgnr2,cnrmr} and \cite{sogge} for a similar class of potentials having Fourier transform in $L^1$) and shown to be bounded $M^p\to M^p$ for every $1\leq p\leq\infty$. As a consequence of Theorem \ref{proimmersioni} we then see that $U(t,s)$ enjoys the continuity property
\[
U(t,s): L^p_k\to L^p,\quad 1<p\leq 2,\quad U(t,s): L^p\to L^p_{-k},\quad 2\leq p<\infty
\]
with $k=k_p=2d|1/2-1/p|$.
\section*{Acknowledgment}
We would like to thank Professors Elena Cordero for several discussions and Kenji Yajima for correspondence about his paper \cite{yajima}.

\end{document}